\documentclass[11pt]{article}
\usepackage{latexsym}
\usepackage{amsfonts,amssymb}
\usepackage{amsmath,amssymb,amsthm,xspace,fullpage}
\usepackage{graphicx}
\newtheorem{theorem}{Theorem}
\newtheorem{lemma}{Lemma}

\usepackage{float}
\usepackage[absolute,overlay]{textpos}
\usepackage[ruled, vlined, linesnumbered]{algorithm2e}
\usepackage{here}

\newcommand{\oracle}{{\mathcal X}}
\newcommand{\suffix}{\mathrm{Suffix}}
\newcommand{\prefix}{\mathrm{Prefix}}
\newcommand{\sibling}{\mathrm{Sibling}}
\newcommand{\siblingL}{\mathrm{SiblingL}}
\newcommand{\parent}{\mathrm{Parent}}
\newcommand{\parentL}{\mathrm{ParentL}}
\newcommand{\labell}{\mathrm{Label}}
\newcommand{\sub}{\mathrm{Sub}}
\newcommand{\query}{\mathrm{Query}}
\newcommand{\yes}{{\it Yes}}
\newcommand{\no}{{\it No}}
\newcommand{\nul}{{\it null}}
\newcommand{\freq}{\mathrm{freq}}
\newcommand{\pos}{\mathrm{pos}}

\def\squarebox#1{\hbox to #1{\hfill\vbox to #1{\vfill}}}

 \newcommand{\bs}{\bigskip} 
  
 \newcommand{\hs}[1]{\hspace*{ #1 mm}} 

\newcommand{\ignore}[1]{}

\allowdisplaybreaks[1]

\begin{document}
\pagestyle{plain}
\begin{center}
{\Large {\bf Reconstructing Strings from Substrings:\\ Optimal Randomized and Average-Case Algorithms}}
\bs\\

{\sc Kazuo Iwama}$^1$ \hspace{5mm} 
{\sc Junichi Teruyama}$^2$ \hspace{5mm} 
{\sc Shuntaro Tsuyama}$^3$ \hspace{5mm} 

\

{\small
$^1${RIMS, Kyoto University, Japan}; \\
{\tt iwama@kuis.kyoto-u.ac.jp} 

$^2${School of Social Information Science, University of Hyogo, Japan}; \\
{\tt junichi\_teruyama@hq.u-hyogo.ac.jp}

$^3${School of Informatics, Kyoto University, Japan}; \\
{\tt stsuyama@kuis.kyoto-u.ac.jp} 

}
\end{center}
\bs

\begin{abstract}
The problem called {\em String reconstruction from substrings} is a
mathematical model of sequencing by hybridization that plays an
important role in DNA sequencing. In this problem, we are given a
blackbox oracle holding an unknown string $\oracle$ and are required to obtain
(reconstruct) $\oracle$ through {\em substring queries} $Q(S)$.  $Q(S)$ is given
to the oracle with a string $S$ and the answer of the oracle is Yes if $\oracle$
includes $S$ as a substring and No otherwise.  Our goal is to minimize
the number of queries for the reconstruction.  In this paper, we deal
with only binary strings for $\oracle$ whose length $n$ is given in advance
by using a sequence of good $S$'s.
In 1995, Skiena and Sundaram first studied this problem and obtained
an algorithm whose query complexity is $n+O(\log n)$.  Its information
theoretic lower bound is $n$, and they posed an obvious open question;
if we can remove the $O(\log n)$ additive term.  No progress has been
made until now.  This paper gives two partially positive answers
to this open question.  One is a randomized algorithm whose query
complexity is $n+O(1)$ with high probability and the other is an
average-case algorithm also having a query complexity of $n+O(1)$ on
average.  The $n$ lower bound is still true for both cases, and
hence they are optimal up to an additive constant.
\end{abstract}

\section{Introduction}

Sequencing by hybridization (SBH)~\cite{DC87, LFKKSM88, PL94} is one of
the major approaches to DNA sequencing which was developed in 1980's
and 1990's.  Its basic idea is to construct, for a given set ${\cal L}$ 
of (short) strings, a single (long) string $T$ that contains all
the strings in ${\cal L}$ as substrings.  Of course a concatenation of
all strings in ${\cal L}$ is a trivial answer. So we usually impose
several constraints for $T$, for instance, $T$ should be shortest or
$T$ should not include another given set of strings.  The problem is
certainly interesting from an algorithmic point of view, but unfortunately,
many nontrivial versions of the problem were proven as intractable
(e.g.,~\cite{GMS80, BPSS02}).

In 1995, Skiena and Sundaram proposed a new approach for SBH which is
more interactive~\cite{SS95}.  Namely we can ask $T$ whether some
string $S$ is its substring or not, {\em sequentially}, i.e., by dynamically selecting
$S$ in each round.  Of course our goal is to select ``good'' $S$'s by
using previously obtained information (yes/no answers of $T$ to the
previous queries).  
More formally they introduced the problem called {\em String 
reconstruction from substrings}.  We are given a black-box oracle
having a hidden string $\oracle$.  If we ask the oracle with a query string 
$S$, then the oracle gives back an answer $\yes$ if $\oracle$
contains $S$ as a substring and $\no$ otherwise.  Our goal is to
reconstruct the string $\oracle$ using a minimum number of
queries.  
This is an idealized model of SBH, having an excellent simplicity as a
mathematical model.

For the case that the alphabet is binary and the length $n$
is known in advance, the authors gave the following 
elegant algorithm, called {\sc SkSu} in this paper. 
(\cite{SS95} also
discusses different cases involving a larger alphabet and/or unknown $n$.
In this paper, however, we are interested in only this binary,
known-$n$ case.)
{\sc SkSu} first obtains the longest 0's, $0^d$, in $\oracle$ using a simple
binary search by spending at most $\log n$ queries (our $\log$ in this
paper is all base-2).  Assume for simplicity
that this $0^d$ appears in $\oracle$ only once.  
Then we ``extend'' this $0^d$ to the right by asking 
if $0^d1$ is a substring (this query is denoted by $\query(0^d1)$).
Note that the answer should be $\yes$ since
$0^d$ is the longest 0's.  Then make $\query(0^d11)$.  If the answer
is $\yes$ then $0^d11$ is confirmed as a substring of $\oracle$.  
Otherwise, $0^d10$ should be a substring since our string is binary.  Repeat
this procedure, namely we add 1 to the current substring and ask the
oracle with that string, until the substring arrives at the right end
of $\oracle$.  We then extend it to the left until its length becomes
$n$ (see the next section for more details).

By using a nice mechanism for detecting the right end, they proved
{\sc SkSu} spends at most $n+\log n +O(1)$ queries and always produces
a correct answer.  The information theoretic lower bound for this
query complexity is $n$, 
and hence, the above upper bound is almost optimal. 
Unfortunately, however, it still has an additive logarithmic gap.  An obvious question is if we can remove
this gap, which is posed as an open question in~\cite{SS95}.

{\bf Our Contribution}.  We give two partially positive answers to this
open question.  Our first algorithm is deterministic which spends
$n+O(1)$ (in fact at most $n+6$) queries {\em on average} and our second one is 
{\em randomized}, which spends
$n+O(1)$ queries with high probability before
it outputs an always correct answer.  It is
straightforward to show that the lower bound for the query complexity
is $n$ for both average case (obvious) and randomized case
(using the Yao's principle).  So our algorithm is
optimal up to an additive constant.

Both algorithms exploit the following fundamental property of binary
strings: Let $S$ be an arbitrarily fixed string of length $\log n$.
Then a constant fraction of strings of length $n$ includes $S$ as
its substring, but this proportion decreases rapidly as the length of
$S$ increases.  For instance, if a string $T$ of length $n$
is randomly selected, it is
unlikely for $T$ to include $S$ of length $\log n +10$.  Similarly, a
constant fraction of $T$ does {\em not} include $S$ of length $\log n$
as its substring, but the proportion decreases rapidly as the length
of $S$ decreases.  Our average-case algorithm is virtually the same as
{\sc SkSu} but we simply include this property in its analysis.

The randomized algorithm is more involved.  The easy case is that the
oracle string $\oracle$ is close to a random string.  Then we can
fully use the above property.  Namely a {\em constant} number of
queries with random strings of length about $\log n$ allow us to find
a substring and a nonsubstring of length about $\log n$ whp, which can
save $\log n$ queries of {\sc SkSu}.  Therefore, we can focus
ourselves on the case that $\oracle$ is far from a random string.
It then turns out that we can define
two groups such that $\oracle$ must belong to either of them.  One
group consists of strings having a lot of repetitions of same substrings.
Intuitively, if we know that a substring $S0$ appears in $\oracle$ but 
$S1$ does not, then we can save one query
whenever we encounter substring $S$.  Thus our profit is large if $S$
repeats a lot.  The other group does not have
many repetitions of same substrings.  Then we can find a ``second seed''
other than the first seed (the longest substring of 0's), which is a bit
longer than the first seed, without any extra queries.  
We can exploit this difference of the length between the first and the
second seeds to remove the $\log n$ gap. 

{\bf Related Work.} 
There are various models for string reconstruction from the
information of substrings and nonsubstrings.
Margaritis and Skiena \cite{MS95} studied the problem called
{\em String reconstruction from substrings in rounds}. 
The query model of this problem is exactly the same as above, but
(two or more) queries can be performed in a single {\em round}. 
Queries in each round can depend on the answers to the queries 
and the answers in the
previous rounds but not on those in the same round and the goal is to minimize
the number of rounds and the number of queries in each round.
Margaritis and Skiena~\cite{MS95} gave a trade-off between 
the number of rounds and the number of queries per round.
There are results about several lower bounds of queries~\cite{FPU99,
Tsur05} in this model, for instance, the lower bound of queries when
the number of rounds is one.
Frieze and Halld\'orsson~\cite{FH02} studied a variant of the model,
in which for each query, the answer is not binary but ternary.  Namely,
it is whether the string appears once in the oracle, appears at least twice, or does not appear.
Tsur~\cite{Tsur10} provided algorithms that improve the results of~\cite{MS95} and~\cite{FH02}.
Acharya et al.~\cite{ADMOP15} investigate the problem deciding whether or not
we can reconstruct the string from substring multisets,
which is a set of frequencies of symbols for all substrings.
Cleve et al.~\cite{C.etal.12} show that if we are allowed to use quantum mechanisms, the query
complexity can be below the classical lower bound or sublinear in $n$.

{\bf Notations}.  We usually use capital 
letters $I,S,T, \ldots$, for strings and
backward small letters $s, t, \ldots$, for symbols.  In this paper,
strings are always binary strings, i.e., a {\em string} $S$ means
$S\in \{0,1\}^*$. Similarly a {\em symbol} $s$ means $s\in \{0,1\}$.
$\overline{0}=1$ and $\overline{1}=0$.
For two strings $S$ and $T$, $ST$ or
$S\!\cdot\! T$ denotes their concatenation.  This includes the case that $S$
and/or $T$ are a single symbol like $S1$ or $S\!\cdot\! 1$.  
For a string $S$, $S[i]$ denotes the $i$-th symbol of $S$.
Let $S=S[1]S[2]\cdots S[m]$.
Then, we call $m$, denoted by $|S|$, the {\it length} of $S$.
$\suffix_i(S)$ denotes string $S[m-i+1]\cdots S[m]$, i.e.,
the {\em suffix} of $S$ of length $i$.  $\prefix_i(S)$ denotes 
string $S[1]\cdots S[i]$, i.e.,
the {\em prefix} of $S$ of length $i$. For a string $S$, $\parent(S)$
denotes $\prefix_{|S|-1}(S)=S[1]S[2]\cdots S[m-1]$, and 
$\sibling(S)$
denotes $\parent(S)\overline{\suffix_1(S)} = S[1]\cdots S[m-1]\overline{S[m]}$. 

We denote our target string in the blackbox oracle by $\oracle$. 
We assume that its length is given in advance and denoted by $n$.
A query to $\oracle$ is denoted by $\query(S)$ and its answer is 1
if $\oracle$ includes the substring $S$ and 0 otherwise.  We
sometimes use $\yes$ for 1 and $\no$ for 0.  For a string $S$, 
$\sub(S)$ denotes the set of all substrings of $S$. 
We often say that a string $Z$ is a {\em substring} of
$S$ if $Z\in \sub(S)$, and a {\em nonsubstring} of $S$ otherwise.
We also say that $Z$ is an $S$-sub ($S$-nonsub, resp.) if $Z$ is a
substring (a nonsubstring, resp.) of $S$.

\section{Basic Algorithm}\label{sec:basic}

\begin{algorithm}[t]
\caption{Procedure {\sc Basic}($S$,$T$)}
\label{alg-basic}
\DontPrintSemicolon
\Indm
\KwIn{ A substring $S$  and a nonsubstring $T$ of $\oracle$}
\KwOut{The oracle string $\oracle$}
\Indp
$i \leftarrow 1$;\;
\While{$i\leq|T|$}{
    \If{$\query(S \!\cdot\! T[1]\cdots T[i-1]\overline{T[i]})=1$}{
    $S \leftarrow S \!\cdot\! T[1]\cdots T[i-1]\overline{T[i]}$; $i \leftarrow 1$;\;
  }\Else{
    $i \leftarrow i+1$;\;
  }
}
$i \leftarrow 1$;\;
\While{$\query(S \cdot T[i])=1$}{
    $S \leftarrow S \cdot T[i]$;
    $i \leftarrow i+1$;\;
}
\While{$|S|<n$}{
    $S \leftarrow$ {\sc ExtendLeft}($S$)\;
}
\Return $S$\;
\end{algorithm}
\begin{algorithm}[!t]
\caption{{\sc ExtendLeft}($S$)}
\label{alg:extendleft}
\DontPrintSemicolon
\Indm
\KwOut{A string $S$.}
\Indp
\lIf{$\query(1S)=1$}{
\Return $1S$
}\lElse{
\Return $0S$
}
\end{algorithm}

The Skiena and Sundaram's algorithm~\cite{SS95} works as follows: For
given $d>0$ such that it is already known that $0^d$ is a substring of
$\oracle$ and $0^{d+1}$ is not, the algorithm extends string $S$
(originally $S=0^d$) to the right by making $\query(S\!\cdot\!1)$.  If
the answer is $\yes$, then $S$ is replaced by $S1$ and by $S0$
otherwise. (If $\query(S\!\cdot\!1)$ is $\no$, the correct symbol
after $S$ should have been $\overline{1}=0$ since our strings are
always binary.)  This extension is obviously correct until the right
end of $\oracle$ comes.  
If the extension has gone beyond the
right end of $\oracle$, all queries after that are answered with
$\no$. In other words, if $\suffix_{d+1}(S)=0^{d+1}$, then we know
that this has happened (recall that $0^{d+1}$ is known to be a
nonsubstring) and there must be the right end somewhere in $0^{d+1}$.
Finding it is easy, i.e., if the current $S$ is $S'0^{d+1}$, then
simply make queries $\query(S'0)$, $\query(S'00)$ and so on 
until the answer becomes $\no$ (if $\query(S'0^j)$
is $\yes$ and $\query(S'0^{j+1})$ is $\no$, then $\oracle$ has a suffix of
$S'0^j$).  Once we have reconstructed the correct suffix, say $S'0^j$,
then all we have to do is to extend it to the left in a way similar to
the above until its length
becomes $n$ which we have assumed is given in advance.

Suppose that the correct suffix is $S'0^j$.  Then the
algorithm has spent $|S'|-d$ queries until the end of $S'$, 
then $d+1$ ones until we have noticed
the right end has been passed, $j+1$ ones to find the right end, and $n-|S'|-j$ ones
for the last phase of left extension, which makes 
\vspace{-9pt}
$$(|S'|-d)+(d+1)+(j+1)+(n-|S'|-j)=n+2\vspace{-9pt}
$$
queries in
total.  Note that we further need queries to obtain the value of $d$
such that $0^d$ is a substring of $\oracle$ and $0^{d+1}$ is not,
i.e., $O(\log n)$ ones in the worst case if we use a simple binary search.

We add a small generalization to this algorithm, by replacing $0^d$ and
$0^{d+1}$ with any (known) substring $S$ and any (known)
nonsubstring $T$ of $\oracle$, respectively.  The new algorithm, 
Algorithm~1, is very
similar:  Our extension to the
right begins from $S$ as before.  Suppose that our current string is
$S$ and the last query (=$\query(S)$) was answered $\yes$. 
Then our next query is
$\query(S\;\overline{T[1]})$ and if the answer is $\no$,
then the next query is $\query(S \; T[1] \;\overline{T[2]})$.  If the answer is again $\no$,
then the next query is
$\query(S\; T[1] \; T[2] \;\overline{T[3]})$, and so on.  If the answer is $\yes$ we
simply confirm
the extension so far, say as $S'$, and restart with
$\query(S'\;\overline{T[1]})$).  Again if the suffix of the current
string becomes $T$ (having $|T|$ consecutive $\no$'s), there should
have been the right end somewhere in this suffix.

We call this generalized
Skiena and Sundaram {\sc Basic}, which will be used in several
occasions in our new algorithms given in the next sections.  Its query
complexity can be obtained exactly as above, giving us our first
theorem, Theorem~\ref{theorem-basic}. 
\begin{theorem}\label{theorem-basic}
{\sc Basic} is correct and its query complexity is $n-|S|+|T|+1$.
\end{theorem}

\section{Average-Case Algorithm}

Our algorithm is simple: We first ask $\oracle$ if $0^{\log n}$ is a
substring. If yes, we ask, for each $i=1,2,\ldots$, if $0^{\log n+i}$
is an $\oracle$-sub until the answer becomes \no.  Otherwise, we ask,
for each $j=1,2,\ldots$, if $0^{\log n-j}$ is an $\oracle$-sub until
the answer becomes \yes.  Thus we can find the integer $d$ (maybe
negative) such that $0^{\log n+d}$ is an $\oracle$-sub but $0^{\log n
  +d+1}$ is not. (For instance, $d=1$ means that $\query(0^{\log n})$
returns \yes, $\query(0^{\log n +1})$ \yes, and $\query(0^{\log n
  +2})$ \no.  Thus we need 3 queries for $d=1$.  Similarly we need
3 queries for $d=-2$.)  Then we call {\sc Basic($0^{\log n +d},0^{\log n
    +d+1}$)}, which gives us the final answer $\oracle$.
Note that  {\sc Basic($0^{\log n +d},0^{\log n+d+1}$)} spends $n+1$ queries regardless
of the value of $d$.
So our goal is to obtain the total number, $N$, of
queries to obtain $d$ in the above procedure for all $2^n$ strings.  
Let $X_i$ be the set of length-$n$ binary strings that 
have substring $0^i$ but not $0^{i+1}$,
and let $f(i)=|X_i|$.  Then $N$ can be written as 
\begin{eqnarray*}
N &=& f(\log n)\cdot 2 +f(\log n+1)\cdot (2+1) +f(\log n+2)\cdot (2+2)+\cdots + f(n) \cdot (2+n-\log n)\\
  && +f(\log n-1)\cdot 2 +f(\log n-2)\cdot (2+1)+f(\log n-3)\cdot (2+2)+\cdots + f(0) \cdot (2+\log n-1)\\
  &=& 2\cdot \left[f(0) + \cdots + f(\log n-2)+f(\log n-1)+f(\log n)+f(\log n+1)+\cdots + f(n)\right]\\
  && +\left[f(\log n)\cdot 0 +f(\log n+1)\cdot 1 +f(\log n+2)\cdot 2 +\cdots + f(n)\cdot (n-\log n)\right]\\
  && +\left[f(\log n-1)\cdot 0 +f(\log n-2)\cdot 1+f(\log n-3)\cdot 2+\cdots + f(0) \cdot (\log n-1)\right].
\end{eqnarray*}

The first sum is obviously $2\cdot2^n$.  For the second sum, 
note that for $\ell\geq \log n +1$,
$f(\ell)+f(\ell+1)+\cdots +f(n)$ is the number of strings that
have substring $0^{\ell}$.  Since $0^{\ell}$ can start from $n-\ell+1$ different
positions in $\oracle$, by using a union bound, we have 
$$f(\ell)+f(\ell+1)+\cdots +f(n)\leq \alpha(\ell) :=(n-\ell+1) \cdot 2^{n-\ell}.$$
Obviously $f(\ell)\leq\alpha(\ell)$ (this approximation is not very bad) 
and it is easy to see that 
$\alpha(\ell+1) \leq \alpha(\ell)/2$.  Therefore, using 
$\sum_{i=0}^{\infty} i \cdot 2^{-i}=2$, the second sum is bounded by
$$\alpha(\log n) \cdot \left( \frac{0}{2^0} + \frac{1}{2^1} + \frac{2}{2^2} + \cdots \right) 
\leq (n-\log n +1)\cdot 2^{n-\log n}\cdot 2 \leq 2\cdot 2^n.$$

For the third sum, note that for $\ell \leq \log n -1$,
$f(\ell)+f(\ell-1)+\cdots +f(1)$ is the number of strings that
do {\em not} have substring $0^{\ell+1}$, which can be bounded as
follows due to \cite{Fu85, CFK95}.
$$f(\ell)+f(\ell-1)+\cdots+f(1) \leq\beta(\ell) :=\left(1 -\frac{1}{2^{\ell+1}} \right)^{n-\ell+1} \cdot 2^n. $$
For a large $n$ and any $\ell \leq \log n -1$, this $\beta(\ell)$ also
decreases exponentially, namely
\[
\frac{\beta(\ell-1)}{\beta(\ell)}
= 
\frac{ \left(1 -\frac{1}{2^{\ell}} \right)^{n-\ell+2} }
{\left(1 -\frac{1}{2^{\ell+1}} \right)^{n-\ell+1}}
= 
\left(1 - \frac{1}{2^{\ell+1}-1}\right)^{n-\ell+1}
\cdot \left(1 -\frac{1}{2^{\ell}} \right)
< e^{-\frac{n-\log n}{n-1}} < 2^{-1}.
\]
Now, exactly as before, the third sum is bounded by
\begin{eqnarray*}
\beta(\log n -1) \cdot \left( \frac{0}{2^0} + \frac{1}{2^1} + \frac{2}{2^2} + \cdots \right)
&\leq& (1-1/2^{\log n})^{n-\log n}\cdot 2^n \cdot 2 \\
&\leq& (1/e)(1-1/n)^{-\log n} \cdot 2^{n+1} \leq 2^n.
\end{eqnarray*}

Thus $N\leq 5\cdot2^n$. Since {\sc Basic} spends $(n+1)2^n$ queries as mentioned before, 
we have
\begin{theorem}
The average complexity for reconstrucion from substrings is at most $n + 6$.
\end{theorem}

\section{Randomized Algorithm}
\begin{figure}[!t]
\begin{center}
  \includegraphics[height=40mm]{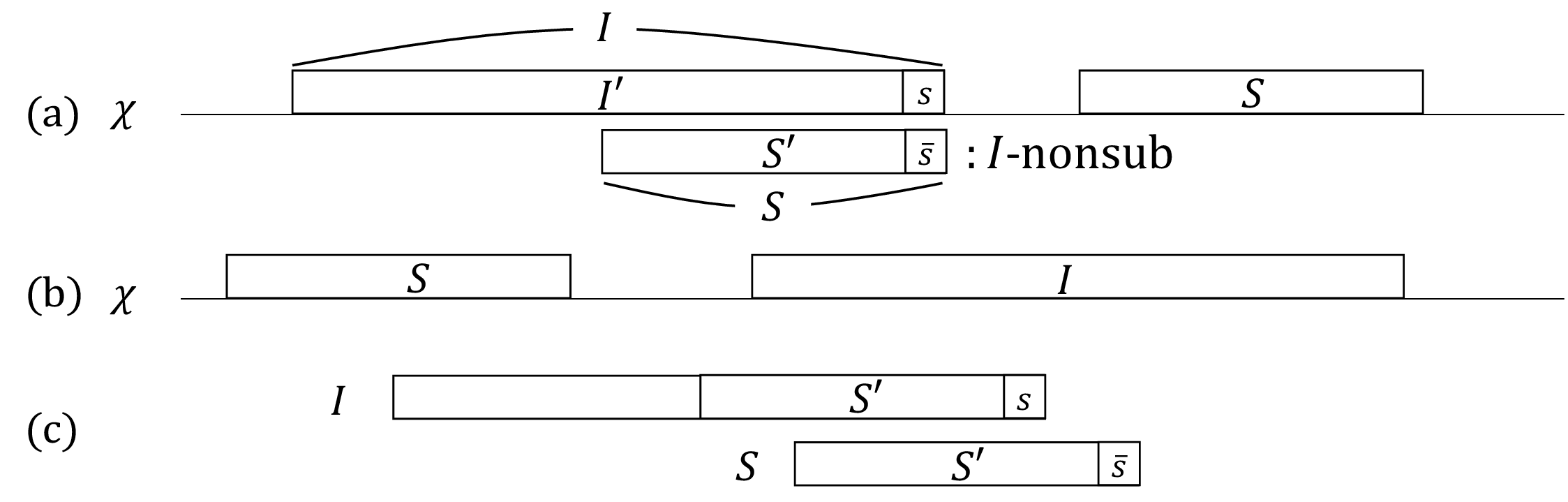}
  \caption{String as a second seed}
\label{secondseed} 
\end{center}
\vspace{-7pt}
\begin{algorithm}[H]
\caption{Procedure {\sc DoubleSeed}}
\label{alg:main}
\DontPrintSemicolon
\Indm
\KwOut{The oracle string $\oracle$}
\Indp
\If{ {\sc TryEasycase}$(Z,d,d_1)=1$}{
  \Return $Z$
}
$q \leftarrow 0.01n$; $\ell \leftarrow d+2d_1$; $r_0 \leftarrow d_1/(2q)$;\;
Assume $\labell(S)$ has value $\nul$ for all strings $S$ initially.\;
$I \leftarrow 0^d$  {\rm \small ($0^d$ is the longest substring of 0's in $\oracle$)}\;
\While{$|I|<q$}{
  \If{$\suffix_{d+1}(I)=0^{d+1}$}{
    \Return {\sc Exception}$(I, d, d_1)$
  } 
  $r \leftarrow$ a uniformly random value in $[0,1]$; {\rm\small
    (Do sample at each position with prob. $r_0$.)}\;
  \If{$r<r_0$}{
        $j \leftarrow 0$;\;
        \While{$\sibling(\suffix_{\ell+j}(I)) \in \sub(I)$}{
        $j \leftarrow j+1$; {\rm\small (extending $S$ in Fig.~\ref{secondseed} to the left until it becomes $I_0$-nonsub)}\;
    }
    \If({{\rm\small (if the position is
single-child)}}){$\query(\sibling(\suffix_{\ell+j}(I))=0$}{
      $\labell(\parent(\suffix_{\ell+j}(I))) \leftarrow$ the last symbol of $I$;\;
    }\Else{\Return {\sc 2ndSeed}$(I, \sibling(\suffix_{\ell+j}(I)))$
      }

  }
    \For{$j:=0$ \KwTo $|I|-\ell$}{
    \If({{\rm\small (if the next position is registered as
single-child)}}){$\labell(\suffix_{\ell+j}(I)) \neq$ \nul}{
        $I \leftarrow ${\sc TwoExtension}$(I, \labell(\suffix_{\ell+j}(I)))$;\;
        {\bf break}
        }
    \If{$j=|I|-\ell$}{
        $I \leftarrow$ {\sc ExtendRight}$(I)$;
        {\rm\small (the same as {\sc ExtendLeft} except the direction)}
        }
    }
}
\Return {\sc Basic}$(I, 0^{d+1})$
\end{algorithm}
\end{figure}

Here are our basic ideas for the algorithm and its analysis. Our main
routine is {\sc\bf DoubleSeed} (see Algorithm~\ref{alg:main}).  As
mentioned in the previous section, if we can find an $\oracle$-sub of
length $\log n$ using a constant number of queries, 
then we are done.  This is exactly what we do in
{\sc\bf TryEasycase} (see Algorithm~\ref{alg:easy})
at the beginning of {\sc DoubleSeed}.  Note
that we seek a substring that is shorter than $\log n$ by a 
constant $C_1$. (This constant is not harmful since our target complexity has
the $O(1)$ term.)  If there
is no such ``easy'' substring, we obtain a seed, $0^d$, such that it is an
$\oracle$-sub but $0^{d+1}$ is not, to start the extension with it.  
Let $d_1$ be the cost for this search.  Then it turns out that $2\log d$
is enough for this $d_1$ by using a sort of binary search (we
want a smaller cost for a smaller $d$).
In {\sc TryEasycase},  we also check if $1^{d+2d_1}$ is an $\oracle$-sub
and if so, we are again done by Theorem 1. 
The reason why we try to find such a 1's (not 0's) substring
is related to the repetition structure of $\oracle$ and will be
stated later.

If {\sc TryEasycase} fails, i.e., if it returns with 0, 
then after setting important values $q$, $\ell$ 
and $r_0$ at Line 3, we start
extending the seed, $0^d$, to the right, until its length becomes $q=0.01n$
(Line 6).
If we can successfully do it, then we have already achieved our goal,
namely we can achieve a query complexity of $n+O(1)$ by continuing
with the standard SkSu at Line 24.  This is our main claim in the following
sections. Let $I_0$ be this extended string of length $q$.  If we reach
the right end before the length becomes $q$, it means we could not
have done what should be done in the main loop because the seed is
located too close to the right end of $\oracle$ (by chance).  So we go to a sort of
exception handling routine, 
{\sc\bf Exception} (see Algorithm~\ref{alg:exception}), at Line 8, and complete our (failed) job in the main
loop with extending the current string to the left.  Note that it is easier to
do this than the original main loop, since we have no chance of encountering
the right end of the oracle.

\begin{algorithm}[!t]
\caption{Procedure {\sc TryEasycase}$(Z,d,d_1)$}
\label{alg:easy}
\DontPrintSemicolon
\Indm
\KwOut{
$1$ with the final answer $Z$ or $0$ with integers $d$ and $d_1$
such that $0^d$ is an $\chi$-sub and $0^{d+1}$ not, and $d_1$ is the
query cost spent in this routine
}
\Indp
\If({{\rm\small  ($C_1$ is a constant related to the error probability) }}){ $\query(0^{\log n - C_1}) = 1$}{
    \While{true}{
      Generate a random string $X$ of length $\log n + C_1$\;
      \If{$\query(X)=0$}{
          {\bf break}
      }
  }
  $Z = ${\sc Basic}$(0^{\log n - C_1}, X)$\;
  \Return $1$
}
$i \leftarrow 0$\;
\If{ $\query(0^{0.5\log n}) = 1$}{
    \While{$\query(0^{\log n - C_1 - 2^i}) = 0$}{
        $i \leftarrow i+1$\;
  }
  $d \leftarrow $ BS$(0^{\log n -C_1 - 2^i}, 0^{\log n -C_1 - 2^{i-1}})$;
  {\rm\small (BS($0^x,0^y$) returns the value $d$ by using the standard binary search 
  that finds the $0^d$ in the range of $0^x$ and $0^y$)}\;
  $d_1 \leftarrow \#$ of oracle queries {{\rm\small ($=(\#$ of queries in $BS)  + i + 2$)}}\;
}
\Else{
    \While{$\query(0^{2^i}) = 1$}{
        $i \leftarrow i+1$\;
  }
  $d \leftarrow $ BS$(0^{2^{i-1}}, 0^{2^{i}})$\;
  $d_1 \leftarrow \#$ of oracle queries\;
}
\If({{\rm\small ($C_2$ is a constant related to the performance)}}){$d_1 \leq C_2$}{
  $Z = ${\sc Basic}$(0^d, 0^{d+1})$\;
  \Return $1$
}
\If{$\query(1^{d+2d_1})=1$}{
  $Z = ${\sc Basic}$(1^{d+2d_1}, 0^{d+1})$\;
  \Return $1$
}
\Return $0$
\end{algorithm}

Each round of the main loop is ``sampled'' with probability $r_0$
(in Lines 9--10, and Lines 11--17 are skipped if not sampled).
Note that ``each step'' is the same as ``each position'' of $I_0$.
Namely when we say ``at each position'' it means we have already
extended the seed up to $I$ that is a prefix of $I_0$, and we are now
looking at the last symbol of $I$.  See Fig.~\ref{secondseed}.  At each
position, we obtain $S=S'\overline{s}$ (Lines 12--13 where no queries are needed) 
that is the shortest string such that (i) its sibling, $S's$, is a
suffix of $I$, (ii) its length is at least $\ell$, and (iii) it
is an $I$-nonsub.  The condition at Line 12 must be
met eventually since at least $\sibling(I) \neq I$. Now we are ready
to introduce the two cases.\\
\hspace*{10mm}(1) $S$ ($=S'\overline{s}$) is an $\oracle$-sub.
\hspace{10mm}(2) $S$ is an $\oracle$-nonsub.\\
If Case (1) happens, substring $S'$ can be
followed by both 0 and 1 in $\oracle$, but only 0 or only 1 follows
$S'$ if Case (2) happens.
Now we define a {\em single-child} position and a {\em
double-child} position for each position $h$ of $I_0$:  If
Case (1) happens, position $h$ is called double-child and otherwise
(if Case (2) happens) called single-child.

Suppose that the current position is double-child and is sampled.
Then we go to {\sc\bf 2ndSeed} (see Algorithm~\ref{alg:2nd}) at Line 17.  Because $S=S'\overline{s}$ is an
$I$-nonsub, $S$ should appear on the right side of $I$ or on its left
side, as shown in Fig.~\ref{secondseed} (a) and (b).  {\sc 2ndSeed} does not know
which side, but it extends this new seed
$S$ to the right in its while loop at Line 2.   If $S$ is on the left side (Fig.~\ref{secondseed} (b)), then its extension hits $I$ and returns at Line 7. Otherwise, the extension hits the the right end of
$\oracle$ and comes to Line 17.  If $I$ and $S$ do not overlap and there is a gap of at least $d$ between them, our task is easy (Lines 17--21):
We simply extend $I$ to the left, get to the left
end and fill the gap between $I$ and $S$ to
obtain the final answer $\oracle$.  Let us look at the
bookkeeping on the cost using the notion of ``profit'' and ``loss:''
Recall that we obtained the first seed of length $d$ using $d_1$
queries, where we count (i) $d$ as a profit and (ii) $d_1$ as a loss.
To know the right and the left ends of $\oracle$ as above,
we have a loss (iii) $d$ for each.  Furthermore, we spend $d_1/2$ expected
queries for the sampling (recall the value of $r_0$), which is not more
than (iv) $d_1$ queries whp (a loss). 
Do not forget that we have obtained the second seed
$S$ of length at least (v) $d+2d_1$ (a profit).  Thus our
profit is (i)+(v) and our loss is (ii)+(iii)$\times 2$+(iv), which 
balance and means we have spent at most $n$ queries in total.
We also need to consider the case that the gap between $I$ and $S$ is
small (the {\bf if} condition at Line 19 is not met and we come to Line 22) and the more messy case that $I$ and $S$ overlap (Fig.~\ref{secondseed} (c)). 
Fortunately it turns out that
we can enjoy a similar balance for all the cases, namely we can prove that
if we go to {\sc DoubleSeed}, we are done.

\begin{algorithm}[!t]
\caption{Procedure {\sc Exception}($I, d, d_1$)}
\label{alg:exception}
\DontPrintSemicolon
\Indm
\KwIn{$I=I'0^{d+1}$ is a string such that $I'0^{m}$ is s suffix of
$\oracle$ for some $m$. $d$ and $d_1$ are the same as those in the main routine}
\KwOut{The oracle string $\oracle$}
\Indp
$q \leftarrow 0.01n$; $\ell \leftarrow d+2d_1$; $r_0 \leftarrow d_1/(2q)$;\;
Assume $\labell(S)$ has value $\nul$ for all strings $S$ initially.\;
$I \leftarrow$ {\sc FindRightend}($I$)\;
$k \leftarrow |I|$\;
\While{$|I|<q+k$}{
    $r \leftarrow$ a uniformly random value in $[0,1]$; {\rm\small
    (Do sample at each position with prob. $r_0$.)}\;
  \If{$r<r_0$}{
        $j \leftarrow 0$;\;
        \While{$\siblingL(\prefix_{\ell+j}(I)) \in \sub(I)$}{
        $j \leftarrow j+1$;\;
    }
    \If{$\query(\siblingL(\prefix_{\ell+j}(I))=0$}{
      $\labell(\parentL(\prefix_{\ell+j}(I))) \leftarrow$ the first symbol of $I$;\;
    }\Else{\Return {\sc 2ndSeed}$(I, \siblingL(\prefix_{\ell+j}(I)))$
      }

  }
    \For{$j:=0$ \KwTo $|I|-\ell$}{
    \If{$\labell(\prefix_{\ell+j}(I)) \neq$ \nul}{
        $I \leftarrow ${\sc TwoExtensionL}$(I, \labell(\prefix_{\ell+j}(I)))$;\;
        {\bf break}
        }
    }
    $I \leftarrow$ {\sc ExtendLeft}$(I)$;
}
\Return {\sc Fill}$(\epsilon, I)$
\end{algorithm}

If the sampled position is single-child, it means we have found
the substring $S'$ such that if $S'$ appears 
anywhere in $\oracle$, its next symbol is always $s$. This
information is kept in the database (at Line 15 of {\sc DoubleSeed}).

Now whether or not the current round is sampled, if we do not go to
{\sc 2ndSeed}, we come to Line 18 of {\sc DoubleSeed}.
Here, if the next position (sometimes denoted as the current position $+1$) is single-child and that information is
already stored in the database, we can make a single
extension without a query. Look at Lines 18--22. Here we
first search a string $S'$ existing in the database 
and go to {\sc\bf TwoExtension} (see Algorithm~\ref{alg:2ext}), which is a bit 
complicated because of the right-end issue and gives us only $1/2$ extension
for free on average.  If we cannot find such $S'$ in the database, then we simply
extend the current string by one at Line 22 using a query.
Note that we need Case (1) (the position is sampled and double-child)
only once to go to {\sc 2ndSeed}.  So,
without loss of generality we can assume we encounter a good number of
single-child positions.  If some of them are sampled and have substrings $S'$ 
that repeat many times in
$\oracle$, we have a corresponding amount of profit, hopefully an
enough one to recover the several losses mentioned above.  
It turns out that it is important for this purpose
that we set value $\ell$ sufficiently smaller than $\log n$.  Details
will be given in Sec.~\ref{sec:analysis}.

Finally we give a short reason for the importance of the sampling.  Suppose
that we have 100 single-child and 100 double-child positions in $I_0$.  Then if
we sample too many, say 100 positions, we would not have any profit since we
need an extra query (Line 12) for each sampled position.  However
if we sample, say, 10 positions at random, this extra cost is
only 10 and about one-half of them should be single-child positions. 
If some of them have strings $S'$ that 
repeats many times in $I_0$, we have a good chance of getting
enough profit.  See the next section for details.

Other routines are {\sc ExtendRight} which is exactly the same as {\sc ExtendLeft} in Sec.~\ref{sec:basic} except the direction of
extension, {\sc FindRightend} that is already explained and given as
Lines 7--9 of {\sc Basic} in Sec.~\ref{sec:basic}, {\sc FindLeftend}
that is similar to
{\sc FindRightend} and {\sc Fill} that fills a gap between the given prefix
and suffix of $\oracle$, by extending the prefix to the right, until
the total length becomes $n$.  We omit 
pseudo codes for those easy routines.

\clearpage
\newpage

\begin{figure}[!t]
\begin{algorithm}[H]
\caption{{\sc 2ndSeed}($I,S$)}
\label{alg:2nd}
\DontPrintSemicolon
\Indm
\KwIn{Two $\oracle$-subs $S$ and $I$ s.t. 
$S$ is an $I$-nonsub and $\sibling(S)=\suffix_{|S|}(I)$. Note that $I$
can be written as $0^d1Z$.
}
\KwOut{The oracle string $\oracle$}
\Indp
$k \leftarrow |S|$\;
\While{{\bf true}}{
  $S \leftarrow$ {\sc RightExtend}$(S)$\;
  \If({{\rm\small ($S$ may hit $I$ or the right end)}}){$\suffix_d(S)=0^d$}{
    \If{$\query(S\!\cdot\!1\!\cdot\!Z[1]$) is \yes}{
      \If{$\query(S\!\cdot\!1\!\cdot\!Z$) is \yes}{
        \Return{\sc Basic}($S\!\cdot\!1\!\cdot\!Z$, $0^{d+1}$)\;
      }
      $S\leftarrow S\!\cdot\!1\!\cdot\!Z[1]$\;
    }
    \Else{
      \If{$\query(S\!\cdot\!1\!\cdot\!\overline{Z[1]}$) is \no}{
        \If{$\query(S\!\cdot\!1)$ is $\yes$}{
              $S \leftarrow S\!\cdot\!1$\;
        }\Else{
            $S \leftarrow$ {\sc FindRightend}$(S)$;\;
        }
        {\bf break}\;
      }
      $S \leftarrow S\!\cdot\!1\!\cdot\!\overline{Z[1]}$\;
        }
    }
}
\While{$|I|+|S| \leq n$}{
  $I \leftarrow ${\sc LeftExtend}($I$)\;
  \If{$\prefix_{d+1}(I)=0^{d+1}$}{
    $I \leftarrow ${\sc FindLeftend}($I$)\;
    \Return {\sc Fill}$(I,S)$\;
    }
}
\If{$\query(I)$ is $\no$}{
  $I \leftarrow ${\sc FindLeftend}($I$)\;
  \Return {\sc Fill}$(I,S)$\;
}
\For{$j \leftarrow k$ \KwTo $1$}{
  \If{$\suffix_j(I)=\prefix_j(S)$ and $\query(I\!\cdot\!\suffix_{|S|-j}(S))=1$}{
    \Return {\sc Fill}$(\epsilon, I\!\cdot\!\suffix_{|S|-j}(S))$
            {\rm\small ($\epsilon$ is the empty string)}
    }
}
\end{algorithm}
\begin{textblock*}{0.4\linewidth}(350pt, 200pt)
\begin{verbatim}
(Let, say, d=5 and let s=Z[1]:
S hits I (line 6) as
S=.........100000
I=          000001s......
Or S continues (line 8) as
S=.........1000001s...
Or S continues (line 16) as
S=.........1000001s'...
Or S hits the right end as 
S=.........1000001 (line 12)
Or S hits the right end
somewhere in its last 00000
(line 14).)


(Detecting overlap (line 25-27).
When |I|+|S|=n IS may be answer as
I=......110110
S=            110110.....
But, in reality they overlap as
I=......110110
S=         110........)
\end{verbatim}
\end{textblock*}
\begin{algorithm}[H]
\caption{{\sc TwoExtension}($I,t$).}
\label{alg:2ext}
\DontPrintSemicolon
\Indm
\KwIn{$t$ is a unique symbol following $I$ in $\oracle$ unless the right end of $I$ is also that of $\oracle$}
\KwOut{New $I$ with two symbols extended or with $0^{d+1}$, an indication of the right end}
\Indp
$s \leftarrow$ 0 or 1 uniformly at random.\;
\lIf{$\query(I\!\cdot\!t\!\cdot\!s$) is $\yes$}{
  \Return $I\!\cdot\!t\!\cdot\!s$
}
\lIf{$\query(I\!\cdot\!t\!\cdot\!\overline{s})$ is $\yes$}{
  \Return $I\!\cdot\!t\!\cdot\!\overline{s}$
}
\If({{\rm\small ($I\!\cdot\!t$ is the right end of $\oracle$)}}){$\query(I\!\cdot\!t)$ is $\yes$}{
    Add $0$'s after $I\!\cdot\!t$ such that its suffix of length $d+1$ is $0^{d+1}$.\;
  }
\Else({{\rm\small ($I$ is the right end of $\oracle$)}}){
    Add $0$'s after $I$ such that its suffix of length $d+1$ is $0^{d+1}$.\;
}
\Return $I$\;
\end{algorithm}
\end{figure}

\clearpage
\newpage


\section{Analysis of the Algorithm}\label{sec:analysis}

As mentioned in the previous section, the basic structure of
{\sc DoubleSeed} is as follows: (1) It first checks if our $\oracle$
is easy
and we are done if so.  Otherwise we obtain an $\oracle$-sub $0^d$ as a seed.  
(2) It extends the seed to the right, where each round is sampled with probability $r_0$.  
(3) If a sampled round is double-child (its position is a double-child position), we go to {\sc 2ndSeed} and we are done.  
(4) If a sample round is single-child, we store the single-child information in the database. 
(5) If the next position has been known to be single-child,
either it is sampled or not, we do not need a query for the
extension. Now we start our detailed analysis.  It is straightforward
to see that {\sc DoubleSeed} always returns with a right answer, so
our job in this section is to bound the number of queries for several
different cases.

\begin{lemma}\label{easycase}
If {\sc DoubleSeed} ends at Line 2, its query
complexity is $n+\max\{2C_1,C_2\}+1$ with failure (i.e., we cannot
achieve the complexity) probability at most $\frac{1}{2^{C_1}}$.
\end{lemma}
\begin{proof}
{\sc TryEasycase} returns with 1 in three cases:  (1) $0^{\log n -C_1}$ is an
$\oracle$-sub.  Then we try to find an $\oracle$-nonsub of length 
$\log n + C_1$. Since the length of $\oracle$ is $n$, the number of different
substrings of a fixed length is at most $n$.  Therefore the
probability that a random string of length $\log n + C_1$ is an
$\oracle$-sub is at most $n/2^{\log n +C_1}=1/2^{C_1}$.
Then using {\sc Basic}, our query complexity is $n+2C_1+1$ by Theorem 1. 
(2) The overhead $d_1$ is smaller than a constant $C_2$. 
Then this overhead is also an overhead of {\sc Basic} and 
our query complexity is $n+C_2+1$ by Theorem 1.
(3) There are an $\oracle$-sub $1^{d+2d_1}$ (string of 1's, not 0's) and an $\oracle$-nonsub $0^{d+1}$.  
Note that the cost of obtaining the former is 1 and $d_1$ for the latter.  
So, our query complexity is even less than $n$ by Theorem 1.
Note that Cases (2) and (3) do not use randomness.
\end{proof}

\begin{lemma}\label{dvalue}
Suppose that {\sc TryEasycase} returns 0.  Then, (i) $C_2\leq d_1\leq 2\log d + 1$, and
(ii) $d+2d_1 \leq \log n -C_1 + 1$.
\end{lemma}
\begin{proof}
Straightforward and omitted.
\end{proof}

We postpone the analysis of {\sc Exception} at the end of this
section.  Recall that its role is to do what we cannot do in
the rest of the main loop since we encounter the end of the string too
early.  So it is easier to see its analysis after finishing that of the whole
parts of the main loop.
Thus, let us move on to the next event of {\sc DoubleSeed}; it goes to 
{\sc 2ndSeed}.  Although it was mentioned above that we are
done, we are actually not since the second seed may overlap with the
current string $I$. Recall that $I_0$ is the extended string after we 
finished the while
loop.  Fix $I_0$ as an arbitrary string of length $q$ and let $h$ be 
a position in
$I_0$ from which {\sc 2ndSeed} is called with $I$ and $S$.  Let $J(h)$ be
the largest integer $j$ such that, at position $h$,  
$\suffix_j(I)=\prefix_j(S)$ and 
$I\cdot \suffix_{|S|-j}(S)$ is an $\oracle$-sub or 0 if such a positive
$j$ does not exists.  Namely $I$ and $S$
overlap with an {\em intersection} of length $J(h)$ and their {\em
  union} is an $\oracle$-sub.  If $J(h)$ is large, it
effectively shortens the second seed and we may not be
able to
achieve the profit we have expected.  Fortunately,
we have the next technical lemma claiming that such a concern is
needless and then 
Lemma~\ref{2seed}, which assures that if $J(h)$ is small, we have a
sufficient profit.

\begin{lemma}\label{overlap}
The probability that {\sc DoubleSeed} samples at a double-child 
position $h$ such that $J(h) > |S|-4d_1$
is at most $2^{-\frac{d_1}{8}}$ (we call such a position a {\em bad position}).
\end{lemma}
\begin{proof}
We first review the basic property of strings.  Suppose that a string
$X$ overlaps with the same $X$ as illustrated in Fig.~\ref{repeat} (1), i.e.,
$X=X'Z=ZX''$ for some $Z$.
Then one can see easily that if the intersection is one-half or more, 
the string should have a repetition structure, namely, we
can write that $X=\sigma^k\sigma'$, where $\sigma'$ is a prefix of
$\sigma$. We call $\sigma$ a {\em block}.  Next suppose that a string $Y$
also has a repetition structure with a block $\gamma$ and that
$k=\max(|\sigma|,|\gamma|)$.  Furthermore, suppose that $X$ and $Y$
overlap with an intersection of length at least $k'=2k$, i.e.,
$\suffix_{k'}(X)=\prefix_{k'}(Y)$.  Then we can claim that both
\begin{figure}[t]
  \begin{center}
   \includegraphics[height=40mm]{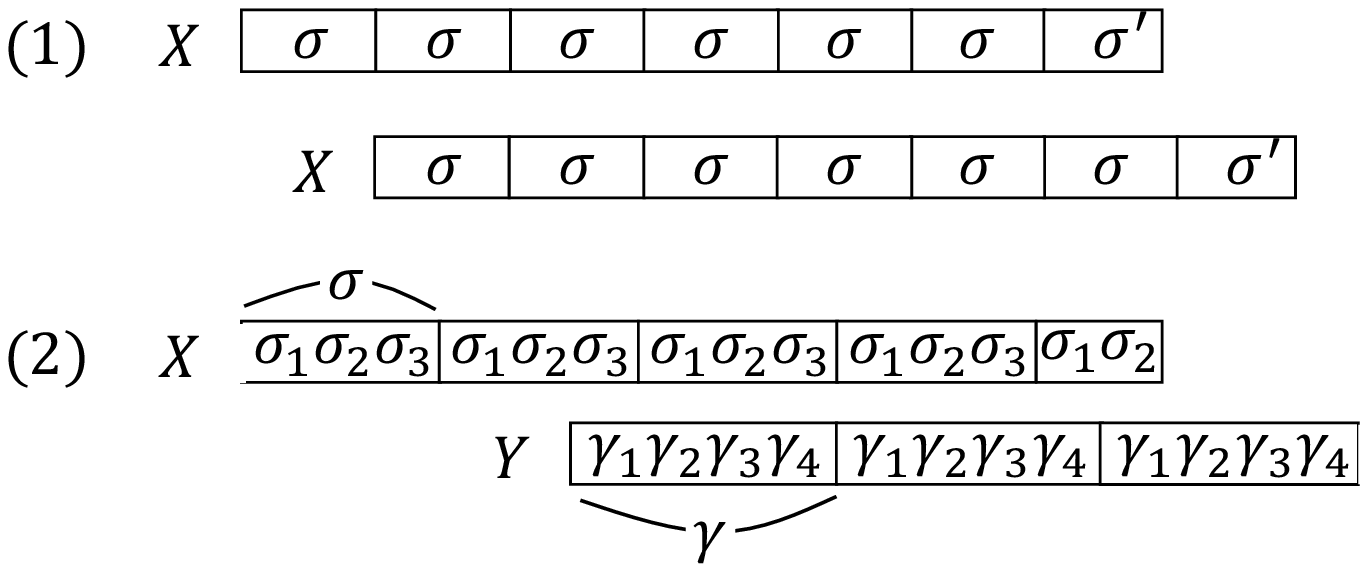}
  \caption{Periodic structure}
  \label{repeat}
 \end{center}
\end{figure}
\begin{figure}[!t]
  \begin{center}
   \includegraphics[height=40mm]{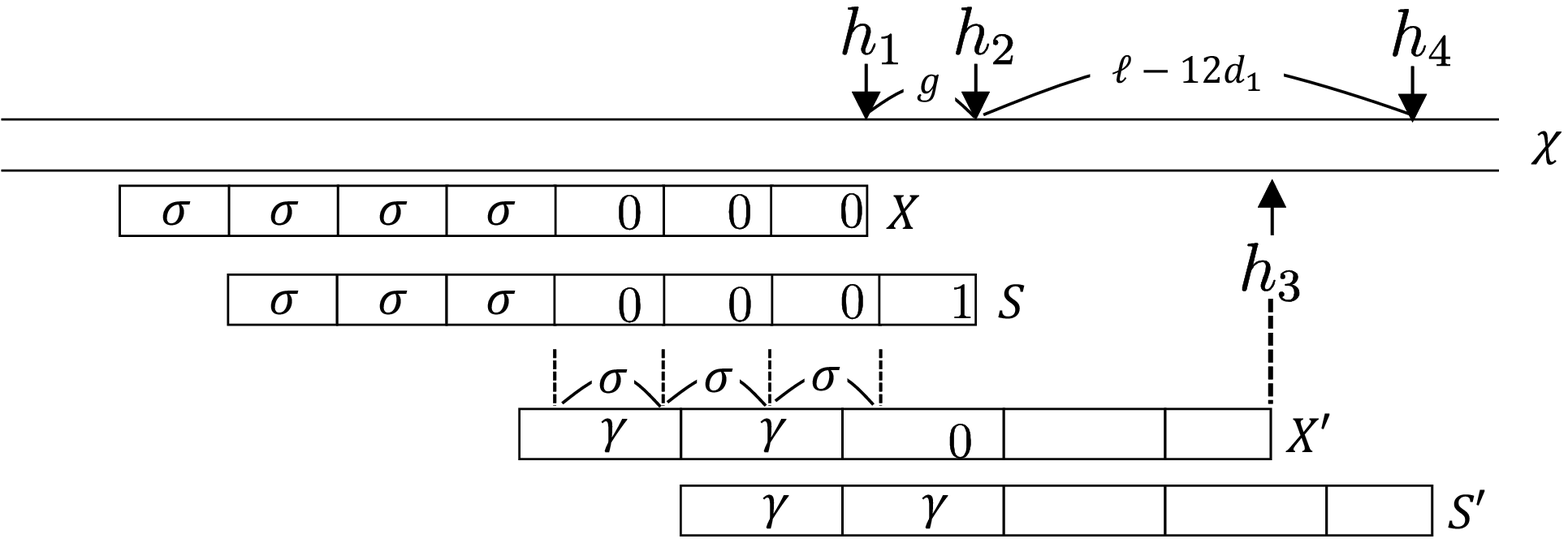}
  \caption{Intervals of bad positions}
  \label{badpositions}
 \end{center}
\end{figure}
strings have the same repetition structure.  (The formal proof may be
a bit messy, but see Fig.~\ref{repeat} (2) for its sketch.  Let $p$ be a common
divisor of $|\sigma|$ and $|\gamma|$.  Then as shown in the figure, we
can write $\sigma=\sigma_1\sigma_2\sigma_3$ and
$\gamma=\gamma_1\gamma_2\gamma_3\gamma_4$ if $|\sigma|=3p$ and
$|\gamma|=4p$, for instance.  Then one can see that
$\gamma_1=\sigma_3=\gamma_4$, $\gamma_2=\sigma_1=\gamma_1$ and
$\gamma_3=\sigma_2=\gamma_2$, implying $\sigma_i$'s and 
$\gamma_j$'s are all identical.  Thus $X$ and $Y$ have the same repetition
structure with the same block.)

Now see Fig.~\ref{badpositions}.  Suppose that current position $h_1$
is a bad position, namely the second seed $S$ and $X=\suffix_{|S|}(I)$
deeply overlap ($h_2$ is the position of the right end of the union).
Their intersection has length $J(h_1) > |S|-4d_1$ by the condition of
the lemma.  Since $S$ and $X$ are the same
strings except for their rightmost symbols, they must have a repetition
structure with a block, say $\sigma$, as mentioned above.  

Because of the large length of the intersection, the length, $g$, between
$h_1$ and $h_2$ is at most $4d_1$.  Consider the position $h_4$ that
is far away from $h_2$ by $\ell-12d_1$.  Then we can show that there is no bad
position between $h_2$ and $h_4$.  Suppose otherwise that $h_3$ (see
the figure) is a bad position.  Then we can define $X'$ and $S'$
exactly as we did for $X$ and $S$ and they heavily overlap with a block $\gamma$
whose length is at most $4d_1$. Furthermore, because both $X$ and $X'$ are
substrings of $\oracle$ and the distance between $h_1$ and $h_3$ is at
most $\ell -8d_1$, $X$ and $X'$ overlap with an intersection of 
length at least $8d_1$.
Thus the two strings satisfy the condition for the above claim,
meaning that their blocks should be
the same.  However,  as shown in the figure, this
implies a contradiction because the symbol in $\oracle$ 
at position $h_2$ cannot be determined uniquely.  
This means that the number of bad positions is less than
$\frac{4d_1}{\ell-8d_1}q$. Hence the probability that these bad
positions are not sampled at all is at least 
$(1-r_0)^{\frac{4d_1q}{\ell-8d_1}} > 1 - \frac{4r_0d_1q}{d-6d_1} = 1 -
\frac{2d^2_1}{d-6d_1}$. Namely the probability that a bad position is
sampled (=the probability of the lemma we want to prove) is at most 
$\frac{2d^2_1}{d-6d_1}$, which is at most
$\frac{2d^2_1}{2^{(d_1-1)/2}-6d_1}$ by Lemma 5.2 (i).  For $C_2\geq
40$ (implying $d_1\geq 40$ by Lemma 5.2 (i)), we have
$2^{(d_1-1)/2}-6d_1 \geq 2^{d_1/2}/2$ and $4d_1^2\leq 2^{(3/8)d_1}$,
which implies $\frac{2d^2_1}{2^{(d_1-1)/2}-6d_1}\leq 1/2^{d_1/8}$.
Thus the lemma is proved.
\end{proof}

\begin{lemma}\label{2seed}
If {\sc 2ndSeed} is called at a position $h$ (called a {\em good
 position}) such that  $J(h) \leq
|S|-4d_1$, then our profit inside {\sc 2ndSeed} is at least $-d+2d_1$.
\end{lemma}
\begin{proof}
We go to {\sc 2ndSeed} with two $\oracle$-subs, $I$ that is the
current extension of the first seed and $S$ a second seed. 
In the while-loop from Line~2, $S$ is extended
right and if the extension continues, we do not have gain or loss as
shown in the small example in the pseudo code
(for instance, we obtain two extensions at
Line 8 with two queries).  We already (in Sec.~4) saw the case that $S$ is located
on the right side of $I$ and they do not overlap with a relatively
large gap. Namely the profit in
{\sc 2ndSeed} is (v)-(iii)$\times 2 =-d+2d_1$.
The case that $S$
is located on the left side of $I$ is rather easy, too.  In 
{\sc 2ndSeed}, we lose $d$ for
detecting the overlap (Line 6) and will lose another $d$ in {\sc
Basic} to detect the right end.  Since our gain is $d+2d_1$, the
overall profit is $-d-d+d+2d_1=-d+2d_1$.  The small-gap case is
similar. 

Now the remaining case is that $S$ is located on the right side of $I$
and they overlap.  Suppose that the length of the intersection 
between $I$ and $S$ is $j_0$ ($\leq |S| - 4d_1$ by the assumption of
the lemma) when
{\sc 2ndSeed} is called, which means if we let $k=|S|$, we have a new profit of
$k-j_0$ at this moment (instead of $d+2d_1$ in the non-overlap case),
since we can think that the first seed is prolonged by $k-j_0$ for
free at Line 26.
However, we have to spend extra queries to determine
this overlap in the for loop of Lines 25--27.  
Here, observe that this check is done downward
from the maximum $k$, i.e., from the possibility that the
intersection is $S$ itself.  Furthermore, if the intersection is more than
one-half of $S$, then as explained in the proof of Lemma~\ref{overlap},
$S$ must have a repetition structure and the length of its
block is at least two (recall that the case that $\oracle$ has a long
1's substring is already excluded using {\sc TryEasycase}).  Thus the
query at Line 26 occurs at most every other round, i.e., the number of
queries is at most $(k-j_0)/2$, assuming that this
query is actually done only if the first condition of the if statement
is met.  If the
intersection is
less than one-half (i.e., $j_0<k/2$), 
the query can start from $k/2$.  Therefore the
number of queries to determine the overlap is at most $k/2-j_0 < (k-j_0)/2$.
Thus  we still keep a profit of $(k-j_0)/2$ in both cases.  
Since $j_0 \leq k-4d_1$ by the assumption of the lemma, 
$(k-j_0)/2$ is at least $2d_1$.  Considering the loss of detecting the
right end (no loss for the left end) our profit inside {\sc 2ndSeed}
is at least $-d+2d_1$.  
\end{proof}

Now we have two remaining cases; {\sc DoubleSeed}'s ending at line 17
and at line 24.  For our analysis, we introduce a threshold for the
number of double-child positions in $I_0$, namely whether it is at least
$q/2^{11}$ or not.  If we have that many double-child positions, 
we can expect at least one of them is sampled at a good position and
we are done by Lemma~\ref{2seed}.  All other possibilities, such as
going to {\sc 2ndSeed} from a bad position and all the double-child
positions have missed being sampled (and ending at line 24), are
regarded as a failure, whose probability turns out to be
sufficiently small.

Conversely suppose that the number of double-child positions is less
than $q/2^{11}$.  Then if we go to {\sc 2ndSeed} (and {\sc DoubleSeed}
ends at line 17), it is still fine; our failure probability due to the bad
positions, etc. is small enough similarly as before.  
However, the probability that we do
not go to {\sc 2ndSeed} (hence eventually 
ending at line 24) becomes a main issue. Namely
we need to assure that we can enjoy a sufficient profit by
free extensions at line 20.  The idea is that the number of different 
strings, $U$, in
the database that give us free extensions is much smaller than the
number of single-child positions in $I_0$.  This is because such a
string $U$ should have a (shorter) prefix that has already appeared in $I$ at a
double-child position (and was not sampled since, if sampled, we should
have gone to the {\sc 2ndSeed}).  Thus the number of different $U$'s
is closely related to the number of double-child positions that is
now assumed to be small.  Thus the same $U$ appears many times, giving
us many free extensions.  Details are given in the proof of the
following lemma.

\begin{lemma}\label{line17or24}
If {\sc DoubleSeed} ends at line 17 or 24,
its query complexity is at most $n$ with failure probability at most $3 \cdot
e^{-\frac{d_1}{2^{13}}}$.
\end{lemma}
\begin{proof}
We consider two cases depending on the number of double-child
positions in $I_0$.

(Case 1: The number of double-child positions in $I_0$ is at least $q/2^{11}$)\hspace{2mm}
First of all, observe Lines 2 and 3 of {\sc TwoExtension} which is
called at Line 20 of {\sc DoubleSeed}.  Suppose neither $I$
nor $It$ hits the right end. Then the condition in Line 2 is met with
probability $1/2$ and if it is not, the condition 
in Line 3 must be met.  So, $I$
is extended by 1.5 symbols on average.
This mechanism means that if a double-child position $P_D$ is 
preceded by a single-child
position $P_S$ and if $P_S$ is extended by
{\sc TwoExtension}, then $P_D$ is skipped regardless of the value of $r_0$
at that round.  However, if this actually happens, then 
{\sc TwoExtension} gives us a profit of $1/2$ per each on average.
Therefore, if a $\delta$ fraction of double-child positions are actually
skipped in this mechanism, we can obtain a profit of $\delta q/2^{12}$, which
is far more than needed (recall that $q$ includes $n$ as a linear form
and what we need
as a profit is a logarithm of $n$) and we are done.
Thus without loss of generality, we can assume most of the $q/2^{12}$ double-child 
positions are subject to being sampled.

Our ``success'' here is only to end at line 17 (if we finish at Line
24 it is counted as failure).
Consider the following three conditions:
(i) No bad position is sampled. (ii) At least one
double-child position is sampled. (iii) The number of samples is
at most $d_1$.  If (i) and (ii) are met, the sample in (ii) forces us
to go to {\sc 2ndSeed} and that sample is not bad.  So,
our profit in {\sc 2ndSeed} is at least $-d+2d_1$ by
Lemma~\ref{2seed}.
If (iii) is met, our profit before coming to {\sc 2ndSeed} is at
least $d-2d_1$, so we have a good balance. Namely if all (i) to (iii)
are met, {\sc DoubleSeed} succeeds.  Now let us calculate the
failure probability. (i) A bad position is sampled with probability at
most $2^{-\frac{d_1}{8}}$ by Lemma~\ref{overlap}.  (ii) The probability that 
there is no sample at any double-child position is at most  $(1-r_0)^{q/2^{12}}< e^{-\frac{r_0q}{2^{12}}} = e^{-\frac{d_1}{2^{13}}}.$  (iii)  {\sc DoubleSeed} spends $d_1/2$ expected queries for sampling. Thus the
probability that this value exceeds $d_1$ is at most
$e^{-\frac{d_1}{6}}$ by Chernoff bound.  The lemma just takes 
the largest one among the three failure probabilities, multiplied by three.

(Case 2: The number of double-child
positions in $I_0$ is less than $q/2^{11}$)\hspace{2mm}
We have further two cases; ending at Line 17 and ending at Line 24.  In the
former case, (i) and (iii) of Case~1 
must be met for the success.  In other words, the (conditional)
failure probability is $P_1=e^{-d_1/8} + e^{-d_1/6} < 2 \cdot e^{-d_1/8}$.  In the following, we prove that if we end at Line 24,
then our (conditional) failure probability is $P_2=3 \cdot
e^{-\frac{d_1}{144}}$, which suffices for the lemma since $P_2$ is
obviously larger than $P_1$.

Fix an arbitrary $I_0$.  If a position $h$ is a single-child
position, we can determine a unique string, $U_h$ of length at 
least $\ell-1$, that is given a
label.  See Fig.~\ref{position} (1) for instance, where $V'$ of length $\ell-1$
appears twice in $I_0$.  Since $V'$ is followed by only 0 in $I_0$ and 
if that is also the case in the entire $\oracle$,
it is given
label 0.  Here, $U_h=V'$. Let
${\cal L} =\{U_h \;|\; h \mbox{ is a single-child position}\}$.
We prove that for any $I_0$, $|{\cal L}|$ is 
much smaller than the number of single-child positions and this means
that many single-child positions have the same string for $U_h$.  We
can assume most of the single-child positions are samplable for
exactly the same reason as in the preceding lemma.  If
$U_h = U_{h'}$, $h'$ comes later than $h$ (see Fig.~\ref{position} (1)), 
and $U_h$ is already labeled (by sampling), then
{\sc TwoExtension} is called at position $h'-1$.  Hence if a same $U_h$ repeats, 
some of them is probably sampled and after that {\sc TwoExtension} is called for every $U_h$.
For the sake of later argument, we divide ${\cal L}$ into two parts, ${\cal L}_1$ collecting
strings in ${\cal L}$ of length $\ell-1$ and ${\cal L}_2$ collecting ones of length 
$\ell$ or more.

Let $V$ be a string
satisfying the following three conditions:
(i) $|V| \geq \ell$,
(ii) there exists a symbol $s \in \{0,1\}$ such that 
$V\!s$ is an $I_0$-sub and $V\overline{s}$ a $I_0$-nonsub, and
(iii) $V'0$ and $V'1$ are both $I_0$-sub, where $V'=\suffix_{|V|-1}(V)$.
Let ${\cal L'}$ be the set of such $V$'s.  We will claim that if
$V \in {\cal L}_2$, i.e., 
if {\sc DoubleSeed} puts the label to a string $V$ and at Line~15 (namely
assuming that the round is sampled) 
and $|V|\geq \ell$, then this $V$ satisfies
these three conditions, i.e., ${|\cal L}_2| \leq |{\cal L'}|$.  (i) and (ii) are obvious and to see (iii) is also
met, look at Fig.~\ref{position} (2).
Suppose that {\sc DoubleSeed} comes to Line 15 at position $h$ 
with $\suffix_{\ell'}(I)=V1$ for $\ell' \geq \ell$,
where $I$ is a prefix of $I_0$ up to position $h$.
For this to happen, the condition of the while-loop at Line~12 should have been 
met for $V'1$ that is the suffix of $V1$ with length one shorter 
(otherwise it would come to Line 15 with the shorter $V'1$).
Namely $\sibling(V'1)=V'0$ is an $I$-sub,
that is, $V'0$ should have appeared before, at position $h'$ in
the figure. Thus (iii) is indeed met.

Thus ${|\cal L}_2| \leq |{\cal L'}|$.  Since the number of
binary strings of length $\ell -1$ is $2^{\ell -1}$, 
we have
$$|{\cal L}|=|{\cal L}_1|+|{\cal L}_2|\leq 2^{\ell -1} + |{\cal
  L'}|.$$
We then bound $|{\cal L'}|$ from above.
By the property (iii),  for each $V \in {\cal L'}$,
there is a double-child position that is the first appearance of $V'0$
or $V'1$ for $V'=\suffix_{|V|-1}(V)$.
(In Fig.~\ref{position} (2), the position  $h'$ can be this double-child position.
If it is sampled then we go to {\sc 2ndSeed}, 
but if not, we can get to the single-child position $h$.)  
Note that two different $V$'s ($0V'$ and
$1V'$) have the same $V'$ and hence for two elements in ${\cal L'}$,
there must be at least one double-child position.  In other words, 
$|{\cal L'}|$ is at most twice the number of double-child
positions, i.e., at most $q/2^{10}$ by the assumption of the lemma.
Therefore, $|{\cal L}|$ is at most $2^{\ell - 1} + q/2^{10} <
q/2^{9}$, 
because by Lemma~\ref{dvalue} (ii), $2^{\ell - 1} \leq \frac{n}{2^{C_1}}$
and this is at most $q/2^{10}$ for $C_1 \geq 20$.

Recall that almost all positions are single-child (exactly speaking it
is at least $(1-2^{-11})q$ by reducing the number of double-child positions,
but there is no harm in using $q$).  For a string $U
\in {\cal L}$, we denote by $\pos(U)$ the sequence of positions $h$ in
$I_0$ such
that $U_h=U$ and let $\freq(U)=|\pos(U)|$. Since $|{\cal L}|
\leq q/2^{9}$ as shown above, $\freq(U)$ is as large as $2^{9}$
on average.  However, $\freq(U)$ can differ arbitrarily for each $U$, which is 
messy for the proof.  So we use the following modest approach: 
$\pos(U)$ is partitioned into blocks of size $\gamma$.  The last block
may be smaller than the others; we call a block of size $\gamma$ a
{\it complete block}.  We ignore a possible profit from smaller blocks and
count only the number of complete blocks for all 
sequences in ${\cal L}$, which is given
as 
$$\sum_{U \in {\cal L}} \left\lfloor \frac{\freq(U)}{\gamma} \right\rfloor \geq \sum_{U \in {\cal L}}\frac{\freq(U)}{\gamma} - |{\cal L}|
> q \left( \frac{1}{\gamma} - \frac{1}{2^{9}}
\right) = \frac{3}{2^8}q$$
by setting $\gamma = 2^7$.

Recall that each position will be sampled with probability $r_0$.  Our
analysis is based on the fact that if the first half of a
complete block is sampled then we can
enjoy {\sc TwoExtension} in the second half. (Note that if some block
is sampled, then the following blocks having the same $U$ will
automatically go to {\sc TwoExtension} in the algorithm.  But in this
(modest) analysis,
new samples are also needed in those following blocks.)  Consider a single
(complete) block.  The probability that its first half is sampled is 
$1 - (1-r_0)^{\frac{\gamma}{2}} > r_0\gamma/2 - O(r_0^2) > \frac{2^7}{3}r_0$ for sufficiently large $n$.
Thus the expected number of sampled blocks is 
$\left(\frac{2^7}{3}r_0\right) \left(\frac{3}{2^8}q\right)=(r_0q)/4=d_1/8$.  Since the sampling is
independent for each block, this value is at least, say $d_1/12$ 
with failure probability $e^{-(\frac{d_1}{8}\cdot(1-\frac{8}{12})^2\cdot \frac{1}{2})}=e^{-\frac{d_1}{144}}$, 
by Chernoff bound.  If a block is sampled, {\sc TwoExtension} is
called $\gamma/2$ times and our profit for each is one with
probability $1/2$.  Thus the expected profit per sampled block is
$\gamma/4=2^5$ and again by Chernoff bound it is at least a $3\cdot 2^3$
fraction with failure probability
$e^{-(\frac{d_1}{12}\cdot2^5\cdot(1-\frac{3\cdot 2^3}{2^5})^2\cdot \frac{1}{2})}=e^{-\frac{d_1}{12}}$.
Thus our total profit is at least $(d_1/12)(3\cdot 2^3)=2d_1$
with the total failure probability $e^{-\frac{d_1}{144}} + e^{-\frac{d_1}{12}}$.

Finally our bookkeeping: Our gain is $d$ for the first seed and $2d_1$ from {\sc TwoExtension} as above.
Our loss is $d_1$ for finding the first seed, $d$ for detecting the right end and
at most $d_1$ for the sampling with failure probability $e^{-\frac{d_1}{6}}$ 
as in the proof of Case~1. 
Thus our gain and loss are balanced with failure probability
$e^{-\frac{d_1}{144}} + e^{-\frac{d_1}{12}}+ e^{-\frac{d_1}{6}} < 3
\cdot e^{-\frac{d_1}{144}}$, which is smaller than the failure
probability stated in the lemma.
\end{proof}

\begin{figure}[tb]
  \begin{center}
  \includegraphics[width=\hsize]{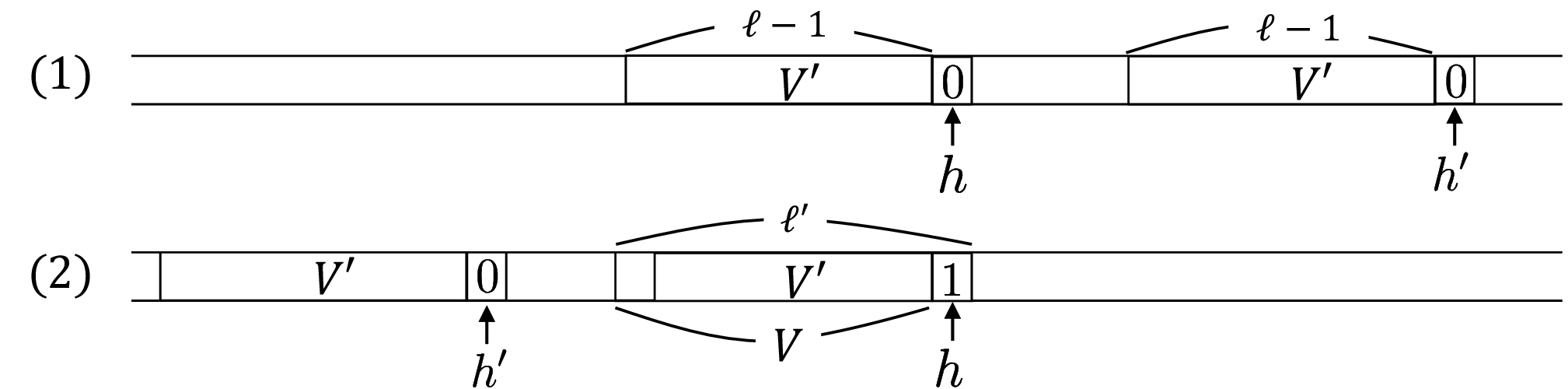}
  \caption{Single-child position and double-child position}
  \label{position}
\end{center}
\end{figure}

Now we are ready to come bask to {\sc Exception}.

\begin{lemma}\label{exception}
If {\sc DoubleSeed} ends at Line 8, its query complexity is at most $n$ 
with failure probability $3 \cdot e^{-d_1/2^{13}}$.
\end{lemma}
\begin{proof}
When it is called, we know $I$
has already gone over the right end of $\oracle$.  So what we do first is to
determine the real right end of $I$ (Line 3).  Then we extend $I$
exactly as the main loop of {\sc DoubleSeed} does.  The amount of
(new) extension is the same, i.e., $q$.  However, the direction of the
extension is to the left, so $\sibling$, $\suffix$ and {\sc
  TwoExtension} are replaced by $\siblingL$, $\prefix$ and {\sc
  TwoExtensionL}, respectively.  The modification is almost obvious;
for a string $S=S[1]S[2]\cdots S[m]$,
$\siblingL(S)=\overline{S[1]}S[2]\cdots S[m]$, and it turns out that 
{\sc TwoExtensionL($I,t$)} can
just return $tI$ without any query, since we do not have to worry
about the left end of $\oracle$ within the while loop.  Note that we
can use {\sc 2ndSeed($I,S$)} as it is, since it covers the case, which
is always the case now, that $S$ exists on the left side of $I$. 

Our bookkeeping analysis in Lemma~\ref{line17or24} is also similar
except the following differences; (i) Our loss to detect the right end
happens at Line 7 of {\sc DoubleSeed} and {\sc Fill} at Line 20 of
{\sc Exception} does not lose anything.  Thus this is neutral.  (ii)
The new {\sc TwoExtensionL} gives us a profit of one instead of $1/2$
before.  Thus our expected total profit is $4d_1$ with the same probability as Case~2 in the proof of Lemma~\ref{line17or24}.
(iii) The loss for the sampling can increase up to $2d_1$ instead of
$d_1$ because we have extra $q$ rounds, 
but this can be compensated for by (ii).
Thus the performance does not become worse than the case that the
algorithm does not go to {\sc Exception}. The failure values are
virtually the same as those in Lemmas~\ref{overlap}--\ref{line17or24}.
\end{proof}

Now we have our main theorem.

\begin{theorem}\label{main}
For any constant $0< \delta \leq 1$, our algorithm spends at most $n + 2^{13}\log_e(3/\delta) + 1$ queries with failure probability at most $\delta$. 
\end{theorem}

\begin{proof}
The procedure {\sc DoubleSeed} ends at Lines~2,~8,~17 and~24.
When it ends at Line~2, 
it spends at most $n + \max\{2C_1, C_2\} + 1$ queries with failure probability $2^{-C_1}$ by Lemma~\ref{easycase}.
When {\sc DoubleSeed} ends at Lines~8,~17 and~24, 
Lemmas~\ref{line17or24} and~\ref{exception} say that
it spends at most $n$ queries 
with failure probability at most $3 \cdot e^{-d_1/2^{13}}$.

Combining these results, {\sc DoubleSeed} spends at most $n + \max\{2C_1, C_2\} + 1$ queries 
with failure probability at most $\max\{ 2^{-C_1}, 3 \cdot e^{-d_1/2^{13}} \}$.
Setting $2C_1 = C_2 = 2^{13}\log_e(3/\delta)$ and using $d_1 \geq C_2$ (see Lemma~\ref{dvalue}),
this failure probability is at most 
\vspace{-4pt}
$$
\max\{ 2^{-C_1}, 3 \cdot e^{-d_1/2^{13}} \} = 
\max\{ 2^{-C_1}, 3 \cdot e^{-C_2/2^{13}} \} = 
\max\left\{ \left( \frac{\delta}{3}\right)^{2^{12}\log_e2}, \delta \right\} = 
\delta.
\vspace{-4pt}
$$
Note that the proof of Lemmas~\ref{overlap} and~\ref{line17or24} require that $C_1 \geq 20$ and $C_2 \geq 40$ hold.
Our setting satisfies these conditions for any $0<\delta\leq1$.
\end{proof}

\section{Final Remarks}

An obvious future work is derandomization of the algorithm of Sec. 4,
although it does not seem easy.  For example, {\sc TryEasycase} is
using randomization to find a nonsubstring of length about $n$ in
constant steps.  There is no obvious way of doing this
deterministically.  Another issue is the time complexity.  In this
paper, we are interested in only the query complexity and in fact our
algorithm seems to spend more than linear number of computation steps,
but this issue should be easier and we already have some idea to make
it linear whp.  The failure probability and the constant term of the
randomized algorithm could be improved, too.


\end{document}